\documentclass[runningheads]{llncs}
\usepackage{graphicx} 
\usepackage{amsmath}
\usepackage{subfigure}
\usepackage[nameinlink,capitalise]{cleveref}
\usepackage{todonotes}
\usepackage{paralist}
\usepackage{lineno}
\usepackage{placeins}

\newtheorem{property2}{Property}

\usepackage{thm-restate}
\graphicspath{{figures/}}

\Crefname{problem}{Problem}{Problem}
\Crefname{figure}{Fig.}{Fig.}

\let\oldendproof\endproof
\renewcommand\endproof{~\hfill$\qed$\oldendproof}

\title{On 1-bend Upward Point-set\\ Embeddings of $st$-digraphs}
\author{Emilio Di Giacomo\inst{1} \and Henry F\"orster\inst{2} \and Daria Kokhovich\inst{3} \and Tamara Mchedlidze\inst{3} \and\\ Fabrizio Montecchiani\inst{1} \and Antonios Symvonis\inst{4} \and Anaïs Villedieu\inst{5}}

\authorrunning{Di Giacomo et al.}
\institute{
University of Perugia
--
  \email {\{emilio.digiacomo,fabrizio.montecchiani\}@unipg.it}
  \and
  Universit\"at T\"ubingen
  --
  \email {henry.foerster@uni-tuebingen.de}
  \and
  Utrecth University
  --
  \email{t.mtsentlintze@uu.nl}
  \and
  National Technical U. of Athens
  --
  \email{symvonis@math.ntua.gr}
  \and
  TU Wien
  -- 
  \email{avilledieu@ac.tuwien.ac.at}
}

\begin{document}
\maketitle

\begin{abstract}
    We study the upward point-set embeddability of digraphs on one-sided convex point sets with at most 1 bend per edge. We provide an algorithm to compute a 1-bend upward point-set embedding of outerplanar $st$-digraphs on arbitrary one-sided convex point sets. We complement this result by proving that for every $n \geq 18$ there exists a $2$-outerplanar $st$-digraph $G$ with $n$ vertices and a one-sided convex point set $S$ so that $G$ does not admit a 1-bend upward point-set embedding on $S$.     
%
\end{abstract}
\section{Introduction}\label{se:intro}

A \emph{point-set embedding (PSE)} of a planar graph $G=(V,E)$ on a given set of points $S$, with $|S|=|V|$, is a planar drawing $\Gamma$ of $G$ such that every vertex of $G$ is represented by a point of $S$ and each edge is drawn as a polyline connecting its end-vertices; if every edge has at most $b \ge 0$ bends, $\Gamma$ is a $b$-bend PSE. 

Gritzmann et al.~\cite{gmpp-eptvs-91} proved that the class of graphs that admit a PSE without bends along the edges on \emph{every} set of points in general position coincides with the class of outerplanar graphs. Efficient algorithms to compute a PSE with no bends on any \emph{given} set of points in general position exist for outerplanar graphs~\cite{b-eogps-02} and trees~\cite{bms-oaetp-97}. Cabello~\cite{JGAA-132} proved that deciding whether a planar graph admits a PSE without bends on a \emph{given} set of points is NP-complete. When bends are allowed, Kaufmann and Wiese~\cite{JGAA-46} proved that every planar graph admits a PSE on every set of points with at most two bends per edge. 

An \emph{upward point-set embedding (UPSE)} of a directed graph $G=(V,E)$ on a given set of points $S$, with $|S|=|V|$,
is a PSE with the additional property that 
each edge $e$ is represented as a polyline monotonically increasing in the $y$-direction;
also in this case we say that $\Gamma$ is a $b$-bend UPSE if every edge has at most $b$ bends. Clearly, for an UPSE to exist $G$ must be an upward planar graph (and thus it must be a DAG).  Different to the undirected case, a characterization of the upward planar digraphs that admit a UPSE without bends on \emph{every} point set is still missing even for points in convex position. On the other hand, Binucci et al.~\cite{DBLP:journals/comgeo/BinucciGDEFKL10} characterize DAGs that admit a 1-bend UPSE on every \emph{upward one-sided convex} (UOSC) point set, i.e., a convex point set such that the bottommost point and the topmost point are adjacent in the convex hull of $S$; the same class has also been characterized by Heath and Pemmaraju~\cite{DBLP:journals/siamcomp/HeathP99} as the class of graphs that admit an upward 1-page book embedding.
For points in convex position Binucci et al.~\cite{DBLP:journals/comgeo/BinucciGDEFKL10} proved that there exist directed trees that do not admit an UPSE on every convex point set and many partial results exists about the embeddability of specific subclasses of directed trees on point sets with different properties~\cite{DBLP:conf/gd/AngeliniFGKMS10,DBLP:conf/walcom/ArsenevaCKM0PV21,DBLP:journals/comgeo/BinucciGDEFKL10,DBLP:journals/comgeo/KaufmannMS13}. Kaufmann et al.~\cite{DBLP:journals/comgeo/KaufmannMS13} studied the problem of deciding whether an upward planar graph admits an UPSE on a given set of points $S$ and show that the problem can be solved in polynomial time for convex point sets, while it is NP-complete for point sets in general position. Arseneva et al.~\cite{DBLP:conf/walcom/ArsenevaCKM0PV21} proved that the problem remains NP-complete even for trees if one vertex is mapped to a specific point. As for the undirected case, two bends per edge suffice for UPSEs of upward planar graphs on any given~set~of~points~\cite{DBLP:journals/jda/GiordanoLMSW15}.  

The results  about (U)PSEs with zero and two bends naturally motivates the study of (U)PSEs with one bend. Testing whether a (upward) planar graph admits a 1-bend (U)PSE is NP-complete in both the upward and the non-upward variants. Indeed, it is easy to see that a 1-bend (U)PSE on a set of collinear points is, in fact, a 2-page (upward) book embedding and deciding whether a (upward) planar graph $G$ admits a $2$-page (upward) book embedding~is NP-complete both in the non-upward~\cite{DBLP:journals/jct/BernhartK79} and in the upward case~\cite{DBLP:journals/tcs/BekosLFGMR23}. However, this relation between 1-bend (U)PSEs and 2-page (upward) book embeddings relies on the use of collinear points
, and thus it does not hold for points in general or in convex position. The following problems are therefore open and worth to investigate. 

\begin{problem}\label{pr:one}
Does every (upward) planar graph admit a $1$-bend (U)PSE on every set of points in general or in convex position?
\end{problem}

\begin{problem}\label{pr:two}
What is the complexity of testing whether a  (upward) planar graph admits a $1$-bend (U)PSE on a given set of points in general or in~convex~position?
\end{problem}

\noindent We study the upward version of \Cref{pr:one} and our contribution is as follows.

\begin{compactitem}
\item On the positive side, we show that every $st$-outerplanar graph (i.e., an outerplanar DAG with a single source and a single sink) admits a $1$-bend UPSE on every UOSC point set (\Cref{th:st-outerplanar}). 
\item We give a negative answer to the upward version of \Cref{pr:one} (\Cref{th:counter-ex}). Namely, we prove that for every $n \geq 18$ there exists a $2$-outerplanar $st$-digraph $G$ with $n$ vertices and an UOSC point set $S$ such that $G$ does not admit an UPSE
on $S$ with at most one bend per edge.
\end{compactitem}

\medskip

\noindent Concerning our second contribution, Di Giacomo et al.~\cite{DBLP:journals/algorithmica/GiacomoDLW06} proved that every two-terminal series-parallel digraph admits a  $1$-bend UPSE on any given set of points. This result has been extended by Mchedlidze and Symvonis~\cite{DBLP:conf/isaac/MchedlidzeS09} to the superclass of $N$-free graphs\footnote{The embedded $N$-graph is shaped like an N, i.e., it contains four vertices $a,b,c,d$ and three edges $(a,b)$, $(c,b)$ and $(c,d)$ such that (1)~$(a,b)$ enters $b$ to the left of $(c,b)$ and (2)~$(c,b)$ exists $c$ to the left of $(c,d)$. An embedded $N$-free graph does not contain the embedded $N$-graph as a subgraph.}. However, there exist $st$-outerplanar graphs that are not $N$-free digraphs (indeed, $st$-outerplanar graphs may contain the forbidden $N$-digraph), and vice-versa. We remark that the study of PSEs is a classical subject of investigation in the Graph Drawing and Computational Geometry literature where different (not necessarily upward) variants have been  studied~\cite{DBLP:journals/jgaa/AngeliniEFKLMTW14,DBLP:journals/tcs/BadentGL08,cr-opse-11,DBLP:journals/ijcga/GiacomoDLMW10,DBLP:journals/tcs/GiacomoGLN20,DBLP:journals/algorithmica/GiacomoLT10,DBLP:journals/comgeo/DujmovicELLLRW13,DBLP:journals/siamdm/DurocherM15,DBLP:journals/jda/FulekT15,DBLP:conf/sofsem/Kaufmann19,DBLP:journals/comgeo/Mchedlidze13,DBLP:journals/gc/PachW01}. In particular, Everett et al.~\cite{DBLP:journals/dcg/EverettLLW10} and Löffler and Tóth~\cite{DBLP:conf/gd/LofflerT15} considered universal point sets for non-upward  1-bend drawings. 

The paper is organized as follows. In \Cref{se:preli} we give preliminary definitions. In \Cref{se:conditions} we prove necessary and sufficient conditions for the existence of a $1$-bend UPSE. In \Cref{se:outerplanar} we describe the construction for outerplanar digraphs, while our negative example is described in \Cref{se:counterexample}. Open problems are in \Cref{se:open}. Proofs marked with ($\star$) are sketched or removed.

\section{Preliminaries}\label{se:preli}

\begin{figure}[tbp]
    \centering
    \subfigure[]{\label{fi:example-a}\includegraphics[width=0.22\textwidth, page=1]{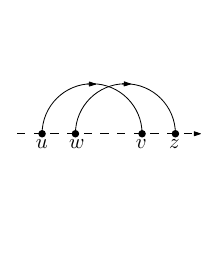}}
    \subfigure[]{\label{fi:example-e}\includegraphics[width=0.22\textwidth, page=5]{figures/example.pdf}}
    \subfigure[]{\label{fi:example-b}\includegraphics[width=0.44\textwidth, page=2]{figures/example.pdf}}
    \centering\subfigure[]{\label{fi:example-c}\includegraphics[width=0.48\textwidth, page=3]{figures/example.pdf}}
    \subfigure[]{\label{fi:example-d}\includegraphics[width=0.44\textwidth, page=4]{figures/example.pdf}}
    \caption{\label{fi:example}
     (a) Two edges that cross; (b) two edges that nest; (c) an example of a 2UBE; (d) an example of a 2UTBE; the bold edges have spine crossings, shown with small crosses; (e) removal of unnecessary sub-edges. }
\end{figure}

Let $G=(V,E)$ be an upward planar graph. A \emph{$2$-page upward book embedding} (2UBE) of $G$ consists of a total order $\prec$ of $V$, that is, a topological sorting of $G$, and of a partition of $E$ into two sets, called \emph{pages}, such that no two edges cross; two edges $(u,v)$ with $u \prec v$ and $(w,z)$ with $w \prec z$ \emph{cross} if the two edges are in the same page and $u \prec w \prec v \prec z$  or $w \prec u \prec z \prec v$ (see \Cref{fi:example-a}). Also, edges $(u,v)$ and $(w,z)$ \emph{nest} if they are on the same page and  $u \prec w \prec z \prec v$ or $w \prec u \prec v \prec z$ (see \Cref{fi:example-e}). We  write $u \preceq v$ if $u$ precedes or coincides with $v$.  A 2UBE can be visualized as an upward planar drawing such that all vertices of $G$ lie along a horizontal line $\ell$, called the \emph{spine}, and each edge is represented as a semi-circle oriented in the direction of the spine and completely contained either above the spine (\emph{top page}) or below the spine (\emph{bottom page}). See \Cref{fi:example-b} for an example of a 2UBE. A \emph{$2$-page upward topological book embedding} (2UTBE) of $G$ is a 2UBE of a subdivision of $G$. 
When considering a 2UTBE as a planar drawing, each subdivision vertex of an edge $e$ can be regarded as a point where $e$ crosses the spine, and therefore is also called a \emph{spine crossing} (see \Cref{fi:example-c}). Further, each of the ``pieces'' of an edge $e$ defined by the subdivision vertices is called a \emph{sub-edge} of $e$; specifically, the sub-edges that are in the top page are called \emph{top sub-edges} and those that are in the bottom page are called \emph{bottom sub-edges}. We write (sub-)edge to mean an element that is either an edge or a sub-edge. We assume that in a 2UTBE no spine crossing has two incident sub-edges that are in the same page; if so, the two sub-edges can be replaced by a single (sub-)edge (see \Cref{fi:example-d}). A 2UTBE is a \emph{single-top 2UTBE} if each edge has at most one top sub-edge (and hence at most two bottom sub-edges). 

A set of points $S$ is an \emph{upward one-sided convex (UOSC) point set} if the points of $S$ are in convex position and the lowest point of $S$ is adjacent to the highest point of $S$ in the convex hull. See \Cref{fi:counter-example-b} for an illustration. We denote by $CH(S)$ the convex hull of $S$. We always assume that all the points of $S$ are to the left of the line passing through the topmost and the bottommost point.
      

\section{Conditions for the existence of a $1$-bend UPSE}\label{se:conditions}

We begin with a necessary condition for the existence of a $1$-bend UPSE.

\begin{restatable}[{$\star$}]{lemma}{UPSEnecessity}\label{le:UPSE-necessity}
Let $G=(V,E)$ be an upward planar graph. If $G$ admits a $1$-bend UPSE on an UOSC point set, then $G$ admits a single-top 2UTBE. 
\end{restatable}
\begin{proof}
	Let $\Gamma$ be a $1$-bend UPSE of $G$ on an UOSC point set $S$. For each edge $e$ of $\Gamma$, replace each intersection point between $e$ and $CH(S)$ that is not an end-vertex of $e$, with a dummy vertex. We obtain a $1$-bend upward planar drawing $\Gamma'$ of a subdivision $G'=(V',E')$ of $G$, such that each edge is drawn completely outside $CH(S)$ or completely inside $CH(S)$. Notice that an edge of $\Gamma'$ that is drawn completely outside $CH(S)$ has necessarily at least one bend, as edges with no bends necessarily lie inside $CH(S)$. Thus, an edge of $\Gamma$ can be split by its intersection points with $CH(S)$ in at most three ``pieces'', at most one of which can be outside $CH(S)$.  
	
	We can define a 2UBE $\gamma'$ of $G'$ as follows. The total order of $V'$ coincides with the bottom to top order of the vertices (real and dummy) in $\Gamma'$; the edges that are drawn inside $CH(S)$ are assigned to the bottom page and those that are drawn outside $CH(S)$ are assigned to the top page. Since $\Gamma'$ is an upward drawing, the total order of $V'$ is a topological sorting. Further, in $\gamma'$, no two edges in the same page cross, as otherwise the same edges would cross in $\Gamma'$. Since $G'$ is a subdivision of $G$, $\gamma'$ is a 2UTBE of $G$. Also, as observed above, each edge $e$ of $G$ is split so that at most one ``piece'' is drawn outside $CH(S)$, hence $e$ has at most one top sub-edge in $\gamma'$.
\end{proof}

Given a 1-bend UPSE $\Gamma$ on an UOSC point set $S$, we say that the 2UTBE $\gamma$ that can be obtained as explained in the proof of \Cref{le:UPSE-necessity} is \emph{induced} by $\Gamma$. 

We now give a sufficient condition for the existence of a $1$-bend UPSE. We begin by introducing some additional definitions and  technical lemmas. Let $\gamma$ be a single-top 2UTBE of an upward planar graph $G$. A sub-edge $(u,v)$ with $u \prec v$ \emph{is nested inside} another sub-edge $(w,z)$ with $w \prec z$ if the two sub-edges are in the same page and $w \preceq u \prec v \preceq z$. 
Notice that it cannot be that $w=u$ and $v=z$ at the same time. An (sub-)edge $(w,z)$ with $w \prec z$  \emph{covers} a vertex $v$ if $w \prec v \prec z$. Let $e_1$ and $e_2$ be two edges of $G$. Edges $e_1$ and $e_2$ form a \emph{forbidden configuration} in $\gamma$ if the following three conditions hold simultaneously: (a) $e_1$ and $e_2$ both have a top sub-edge, one of the top sub-edges is nested inside the other, and the two sub-edges can possibly share a vertex; (b) $e_1$ and $e_2$ both have a bottom sub-edge, one of the bottom sub-edges is nested inside the other, and the sub-edges do not  share a vertex; and (c) each bottom sub-edge covers at least one vertex and each top sub-edge covers at least two vertices. We have four possible forbidden configurations: \emph{Type 1 forbidden configuration} is such that the top sub-edges do not share a vertex and the two bottom sub-edges precede the two top sub-edges in the direction of the spine  (see \Cref{fi:forbidden-a}); \emph{Type 2 forbidden configuration} is like the Type 1 forbidden configuration but with the top edges that share a vertex  (see \Cref{fi:forbidden-b}). \emph{Type 3} and \emph{Type 4 forbidden configurations} are like Type 1~and~Type~2~respectively, but the top sub-edges precede the bottom sub-edges  (see \Cref{fi:forbidden-c,fi:forbidden-d}). We say that the 7 (or 6) vertices necessary to have a forbidden configuration are \emph{the vertices that define} the forbidden configuration. These are the 4 (or 3) end-vertices~of~the~two~edges forming the forbidden configuration and the three vertices that are covered by their sub-edges.
A single-top 2UTBE is \emph{nice} if it has no forbidden configuration.

\begin{figure}[t]
	\centering
	\subfigure[Type 1]{\label{fi:forbidden-a}\includegraphics[width=0.24\textwidth, page=1]{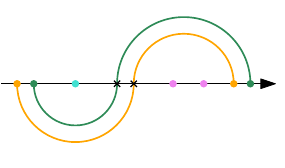}}
	\subfigure[Type 2]{\label{fi:forbidden-b}\includegraphics[width=0.24\textwidth, page=2]{figures/forbidden.pdf}}
	\subfigure[Type 3]{\label{fi:forbidden-c}\includegraphics[width=0.24\textwidth, page=3]{figures/forbidden.pdf}}
	\subfigure[Type 4]{\label{fi:forbidden-d}\includegraphics[width=0.24\textwidth, page=4]{figures/forbidden.pdf}}
	\caption{\label{fi:forbidden.2}
		Forbidden configurations.}
\end{figure}

\begin{figure}[t]
	\centering
	\subfigure[Type 1]{\label{fi:forbidden-e}\includegraphics[width=0.48\textwidth, page=5]{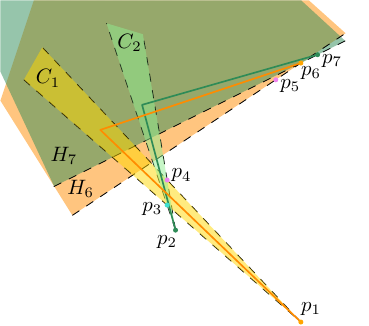}}
	\subfigure[Type 2]{\label{fi:forbidden-f}\includegraphics[width=0.48\textwidth, page=6]{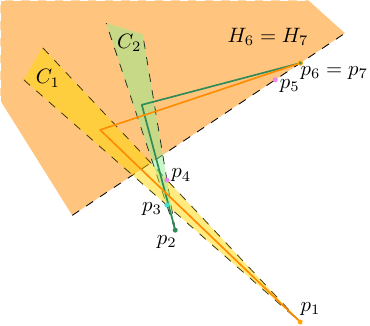}}
	\subfigure[Type 3]{\label{fi:forbidden-g}\includegraphics[width=0.48\textwidth, page=7]{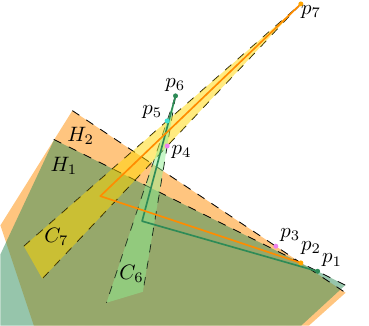}}
	\subfigure[Type 4]{\label{fi:forbidden-h}\includegraphics[width=0.48\textwidth, page=8]{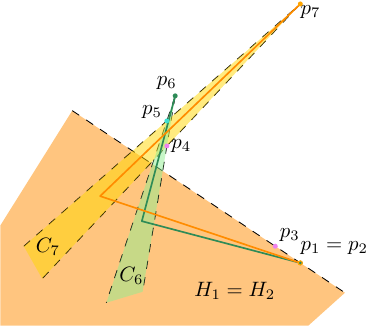}}
	\caption{\label{fi:forbidden.3}
		Impossible point sets.}
\end{figure}


The next lemma shows that forbidden configurations are obstacles to the existence of a 1-bend UPSE for specific set of points. We describe 4 types of UOSC point sets, one for each type of forbidden configuration. See \cref{fi:forbidden.3} Let $p_1, p_2,p_3,p_4,p_5,p_6,p_7$ be a set $S$ of points ordered from bottom to top. Denote by $C_i$, with $i \in \{1,2\}$ the cone defined by the two half-lines starting at $p_i$ and passing through $p_3$ and $p_4$, respectively. Also, denote by $H_i$, with $i \in \{6,7\}$ the half plane above the straight line passing through $p_{i-1}$ and $p_i$. Finally, denote by $T_1$ the portion of $C_1$ that does not intersect $H_6$ and by $T_2$ the portion of $C_2$ that does not intersect $H_7$. We say that $T_1$ and $T_2$ cross each other if every segment connecting $p_1$ to the opposite side of $T_1$  crosses every segment that connects $p_2$ to the opposite side of $T_2$. If $S$ is such that $T_1$ and $T_2$ cross, we say that $S$ is a \emph{Type 1 impossible point set} (see \Cref{fi:forbidden-e}). A \emph{Type 2 impossible point set} is like a Type 1 impossible point set, but with $p_6$ and $p_7$ coincident -- in this case the two half planes $H_6$ and $H_7$ are also coincident (see \Cref{fi:forbidden-f}). \emph{Type 3} and \emph{Type 4 impossible points sets} are like Type 1 and Type 2 respectively, but mirrored vertically (see \Cref{fi:forbidden-g,fi:forbidden-h}). 

\begin{lemma}\label{le:forbidden}
	If a single-top 2UTBE $\gamma$ contains a forbidden configuration of Type $i$, with $i \in \{1,2,3,4\}$, then there does not exist a 1-bend UPSE whose induced 2UTBE is $\gamma$ and such that the vertices that define the forbidden configuration are mapped to an impossible point set of Type $i$. 
\end{lemma}
\begin{proof}
	Assume that $\gamma$ has a Type 1 forbidden configuration (the other cases are analogous). Denote the two edges forming the forbidden configuration as $e_1$ and $e_2$, with the top sub-edge of $e_1$ nested inside the top sub-edge of $e_2$. Suppose that an UPSE exists whose induced 2UTBE is $\gamma$ and such that the vertices of the forbidden configuration are mapped to the points of a Type 1 impossible point set. Then both $e_1$ and $e_2$ have one bend; the bend of $e_1$ is a point of $C_1 \cap H_6$, and the one of $e_2$ is a point of $C_2 \cap H_7$. This implies that the portion of $e_1$ drawn inside $T_1$ crosses the portion of $e_2$ drawn inside $T_2$ (see \Cref{fi:forbidden-a,fi:forbidden-e}).%
%
\end{proof}

In the rest of this section, we prove that if $G$ has a nice single-top 2UTBE then it admits a 1-bend UPSE on every UOSC point set. Let $S$ be an UOSC point set of size $n$. Let $\gamma$ be a 2UTBE of an $n$-vertex upward planar graph $G$ and let $v_1, v_2, \dots,v_{n'}$ be the sequence of vertices along the spine obtained by replacing each spine crossing with a dummy vertex. An \emph{enrichment of $S$ consistent with $\gamma$} is an UOSC point set $S'$ such that: (i) $S \subset S'$; (ii) $|S'|=n'$; and (iii) if we denote by $p_1, p_2, \dots, p_{n'}$ the points of $S'$ in bottom-to-top order, then $p_i$ is a dummy point if and only if $v_i$ is a dummy~vertex.~See~\Cref{fi:construction}.

Let $\gamma$ be a single-top 2UTBE of an upward planar graph $G$, and let $\gamma'$ be the 2UBE obtained by replacing the spine crossings of $\gamma$ with dummy vertices and let $\gamma'_{top}$ be the 1-page book embedding obtained by $\gamma'$ considering only the top page; we call $\gamma'_{top}$ the \emph{top-reduction} of $\gamma$. See \Cref{fi:construction-b}. 
Let $S$ be an $n$-point one-sided convex point set and let
$S'=\langle p_1,p_2,\dots,p_{n'} \rangle$ be an enrichment of $S$ consistent with $\gamma$. We assign to each dummy vertex $v_i$ in $\gamma'_{top}$ a slope $\sigma$, which has to be used to draw the segment incident to the dummy vertex. If $v_i$ is adjacent to a vertex $v_j$ (real or dummy) with $j > i$, then $\sigma$ is a slope of the II-IV quadrant defined by the Cartesian axes, while if $v_i$ is adjacent to a vertex $v_j$ (real or dummy) with $j < i$, then $\sigma$ is a slope of the I-III quadrant. In either case the value of $\sigma$ must be smaller, in absolute value, than the slope of any segment $\overline{p_kp_{k+1}}$ for $k=i,i+1,\dots,j-1$ if $i < j$ or for $k=j,j+1,\dots,i-1$ if $i > j$. Such a choice of slopes is called a \emph{slope assignment} for $\gamma'_{top}$. Let $e_1$ and $e_2$ be two sub-edges with at least one dummy end-vertex each and such that $e_2$ is nested inside $e_1$. The slope assignment  is \emph{good for $e_1$ and $e_2$} if for any two slopes $\sigma_1$ assigned to $e_1$ and $\sigma_2$ assigned to $e_2$ in the same quadrant we have~$|\sigma_1| < |\sigma_2|$. The slope assignment is \emph{good} if it is good for every pair of nested sub-edges.

\begin{figure}[t]
    \centering
    \subfigure[]{\label{fi:construction-a}\includegraphics[width=0.4\textwidth, page=1]{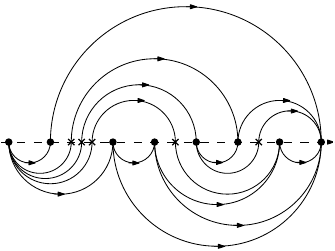}}
    \hfil
    \subfigure[]{\label{fi:construction-b}\includegraphics[width=0.4\textwidth, page=2]{figures/construction.pdf}}
    \hfil
    \subfigure[]{\label{fi:construction-c}\includegraphics[width=0.4\textwidth, page=3]{figures/construction.pdf}}
    \hfil
    \subfigure[]{\label{fi:construction-d}\includegraphics[width=0.4\textwidth, page=4]{figures/construction.pdf}}
    \caption{\label{fi:construction}
     (a) A single-top 2UTBE $\gamma$; (b) the top-reduction $\gamma'_{top}$ of $\gamma$; (c) an enrichment of an UOSC point set $S$ (black squares) consistent with $\gamma$ with a good slope assignment. (d) A $1$-bend UPSE of $\gamma'_{top}$ on $S'$ computed as in  \Cref{le:nice-assignment}. }
\end{figure}

\begin{restatable}[{$\star$}]{lemma}{niceAssignment}\label{le:nice-assignment}
Let $G$ be an $n$-vertex upward planar graph, let $S$ be an UOSC point set. Let $\gamma$ be a single-top 2UTBE of $G$ and let $\gamma'_{top}$ be the top-reduction~of~$\gamma$. If a good slope assignment is given, then $\gamma'_{top}$ has a 1-bend UPSE on every enrichment $S'$ of $S$ consistent with $\gamma$ such that all the edges are drawn outside $CH(S')$ and the segment incident to each dummy vertex is drawn with~the~assigned~slope.
\end{restatable}

\begin{proof} Let $v_1, v_2, \dots,v_{n'}$ be the vertices in $\gamma'_{top}$ according to the spine order. Let $S'=\langle p_1,p_2,\dots,p_{n'} \rangle$ be an enrichment of $S$ consistent with $\gamma$. (See \Cref{fi:construction-c}).
	
	The edges are drawn according to an order defined by the nesting, i.e., an edge is drawn only after that all edges nested inside it are already drawn. Let $e=(v_i,v_j)$ be the current edge to be drawn and suppose that $i < j$, i.e., that $v_i \prec v_j$. The edge $e$ is drawn as the union of two segments: $s_i$ incident to $v_i$ and $s_j$ incident to $v_j$. 
	The segment $s_i$ is drawn in the II quadrant, while segment $s_j$ is drawn in the III quadrant. This guarantees that $s_i$ and $s_j$ can meet at a bend point.
	If $v_i$ (resp. $v_j$) is a dummy vertex, then $s_i$ (resp. $s_j$) is drawn with the slope assigned to $v_i$ (resp. $v_j$). Notice that the slope assigned to $v_i$ (resp. $v_j$) is a slope of the II-IV quadrant (resp. I-III quadrant). If $v_i$ (resp. $v_j$) is a real vertex, then $s_i$ (resp. $s_j$ ) is drawn with a slope $\sigma$ of the II-IV quadrant (resp. I-III quadrant) and such that the value of $\sigma$ is smaller, in absolute value, than the value of any other slope used in the edges nested inside $(v_i,v_j)$ (which have already been drawn). If no edge is nested inside $(v_i,v_j)$, then $|\sigma|$ has to be smaller than the absolute value of the slope of any segment $\overline{p_kp_{k+1}}$, for $k=i,i+1\dots,j-1$. 
	
	The slopes used to draw the edges are such that all the segments are drawn outside $CH(S)$ except for an endpoint for each segment, which coincides with a point of $S'$. This is true because every segment $s$ of an edge is drawn with a slope (assigned or chosen by the algorithm) that is smaller, in absolute value, than the slope of any segment of $CH(S)$ that can potentially intersect $s$. Moreover, the slopes of the two segments of an edge $e$ are smaller, in absolute value, than the slopes of all the segments of the edges nested inside $e$. This implies that there is no edge crossing.
\end{proof}

\begin{restatable}[{$\star$}]{lemma}{leCoverOne}\label{le:cover-1}
	Let $G$ be an $n$-vertex upward planar graph, let $\gamma$ be a single-top 2UTBE of $G$, and let $e$ be a top sub-edge that covers exactly one vertex and that has no top sub-edge nested inside it. Let $\gamma'$ be the 2UTBE obtained from $\gamma$ by removing the edge $e'$ containing the sub-edge $e$. Let $\Gamma'$ be a $1$-bend UPSE of $G \setminus \{e'\}$ on an UOSC point set $S$ whose induced 2UTBE is $\gamma'$. Then it is possible to construct a 1-bend UPSE of $G$ on $S$ that has $\Gamma'$ as a sub-drawing.
\end{restatable}

\begin{proof}
	Let $v_i$ and $v_j$ be the two end-vertices of the edge $e'$ and suppose that $i < j$, i.e., that $v_i \prec v_j$. Let $v_k$ be the single vertex covered by the top sub-edge $e$ of $e'$; clearly, $i<k<j$, i.e., $v_i \prec v_k \prec v_j$. There is no edge that crosses the segments $\overline{v_iv_k}$ and $\overline{v_kv_j}$ (although one of these segments can coincide with an edge) as otherwise there would be an edge crossing $e'$ in $\gamma$. Since no top sub-edge is nested inside $e$, we can draw the edge $e'$ by choosing a bend point slightly above $v_k$ and connecting it with $v_i$ and $v_j$ (see \Cref{fi:construction-2}).  
\end{proof}

\begin{figure}[t]
	\centering
	\subfigure[]{\label{fi:construction-e}\includegraphics[width=0.48\textwidth, page=5]{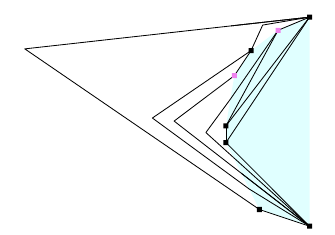}}
	\subfigure[]{\label{fi:construction-f}\includegraphics[width=0.48\textwidth, page=7]{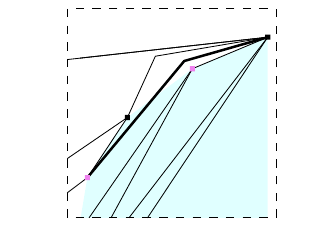}}
	\caption{\label{fi:construction-2}
		(a) A $1$-bend UPSE of the graph of \Cref{fi:construction-a} with one edge removed (the two end-vertices of the removed edge are highlighted); the top sub-edge of the removed edge covers only one vertex; (b) the addition of the missing edge (close-up). }
\end{figure}

\begin{restatable}[{$\star$}]{lemma}{UPSEsufficiency}\label{le:UPSE-sufficiency}
Let $G=(V,E)$ be an $n$-vertex upward planar graph. If $G$ admits a nice single-top 2UTBE, then $G$ admits a $1$-bend UPSE on every UOSC point set $S$ of size $n$.
\end{restatable}
\begin{proof}
	If $G$ admits a nice single-top 2UTBE $\gamma$, then we can compute a 1-bend UPSE on every one-sided convex point set $S$ as follows. Let $S'=\langle p_1,p_2,\dots,p_{n'} \rangle$ be an enrichment of $S$ consistent with $\gamma$.
	We remove all edges that have a top sub-edge covering only one vertex and nest no edges  inside. If after such a removal, some new edges with the same properties are created they are recursively removed, until no more edge exists with the same properties. The reason we remove these edges is to guarantee that, in the description that follows, we can assume that if a pair of edges satisfies Condition (a) of the definition of forbidden configuration, they also satisfy Condition (c), and therefore they do not satisfy Condition (b). The removed edges will be reinserted at the end in reverse order of removal using \Cref{le:cover-1}. 
	
	Let $\gamma'$ be the single-top 2UTBE resulting from the edge removal explained above and let $G'$ be the corresponding graph. We now compute a $1$-bend UPSE of $G'$ on $S'$. We first map each vertex $v_i$ to the point $p_i$ ($i=1,2,\dots,n'$). Notice that by the choice of the additional points, the dummy vertices are mapped to the dummy points. We then draw the (sub-)edges that are in the bottom page as a straight-line segments inside the convex hull $CH(S')$ of $S'$. Since the bottom-to-top order of the vertices along $CH(S')$ is the same as in $\gamma'$, the (sub-)edges drawn inside $CH(S')$ do not cross each other. 
	
	Now, in order to draw the top (sub-)edges, we consider the top restriction of $\gamma'$, and define a slope assignment, assigning to each dummy vertex $d$ the slope of the segment incident to $d$ that is in the bottom page (drawing the segment incident to $d$ with this slope guarantee that no additional bend is created at $d$). We now prove that this slope assignment is good and thus by \Cref{le:nice-assignment} all the top (sub-)edges can be drawn outside the convex hull respecting the slope assignment, which guarantees that each edge is drawn with one bend. 
	
	Let $e_1=(v_{i_1},v_{j_1})$, with $i_1 < j_1$, and $e_2=(v_{i_2},v_{j_2})$ with $i_2 < j_2$ be two top sub-edges such that $e_2$ is nested inside $e_1$, i.e., such that $v_{i_1} \preceq v_{i_2} \prec v_{j_2} \preceq v_{j_1}$. Each sub-edge has one or two assigned slopes depending on the number of dummy vertices. Consider any two slopes, $\sigma_1$ assigned to an end-vertex of $e_1$ and $\sigma_2$ assigned to an end-vertex of $e_2$. In order to prove that the described slope assignment is good we have to prove that either $\sigma_1$ and $\sigma_2$ are in different quadrants, or they are in the same quadrant and $|\sigma_1|<|\sigma_2|$. Thus, in the following, assume that $\sigma_1$ and $\sigma_2$ are in the same quadrant.

	If $\sigma_1$ is assigned to $v_{i_1}$ and $\sigma_2$ is assigned to $v_{j_2}$ or $\sigma_1$ is assigned to $v_{j_1}$ and $\sigma_2$ is assigned to $v_{i_2}$, then they are in different quadrants. So, assume that $\sigma_1$ is assigned to $v_{i_1}$ and $\sigma_2$ is assigned to $v_{i_2}$ (the case when $\sigma_1$ is assigned to $v_{j_1}$ and $\sigma_2$ is assigned to $v_{j_2}$ is symmetric). Consider the bottom sub-edge $e'_1=(v_{k_1},v_{i_1})$ that shares $v_{i_1}$ with $e_1$ and the bottom sub-edge $e'_2=(v_{k_2},v_{i_2})$ that shares $v_{i_2}$ with $e_2$. By Condition (b) of the definition of forbidden configuration, either $v_{k_1}=v_{k_2}$ or the bottom sub-edges are not nested, i.e., the order of the vertices is $v_{k_1} \prec v_{i_1} \prec v_{k_2} \prec  v_{i_2} \prec v_{j_2} \preceq v_{j_1}$. In the first case, the two segments that represent $e'_1$ and $e'_2$ (and that define the slope assigned to $e_1$ and to $e_2$) share an end-vertex and since $v_{i_2}$ follows $v_{i_1}$ the slope of the segment $\overline{v_{k_1}v_{i_1}}$ is smaller, in absolute value than the slope of $\overline{v_{k_1},v_{i_2}}$. In the second case, the segment $\overline{v_{k_2},v_{i_2}}$ has both end-vertices after the end-vertices of the segment $\overline{v_{k_1},v_{i_1}}$; thus from the one-sided convexity of $S$ it follows that the slope of $\overline{v_{k_2},v_{i_2}}$ is larger, in absolute value, than the slope of $\overline{v_{k_1},v_{i_1}}$. This concludes the proof that the slope assignment is good and therefore, by \Cref{le:nice-assignment}, that it is possible to compute a 1-bend PSE of $G'$ on $S'$. 
	
	It remains to re-insert all the edges removed at the beginning of the algorithm, i.e., those that had a sub-edge covering a single vertex and did not have an edge nested inside them. These edges will be re-inserted in the reverse order of removal. As a consequence, each time we add one such edge $e$ we are in the hypothesis of \Cref{le:cover-1} and  it is possible to re-insert $e$. Once all the edges of $G \setminus G'$, are reinserted and the dummy vertices and points are removed, the resulting drawing is a $1$-bend UPSE of $G$ on $S$.
\end{proof}

\section{1-bend UPSE of $st$-outerplanar graphs}\label{se:outerplanar}

\newcommand{\mi}{\textrm{mid}}
A graph is \emph{outerplanar} if it admits an \emph{outerplanar drawing}, i.e., a planar drawing in which all vertices belong to the boundary of the outer face, which defines an \emph{outerplanar embedding}. Unless otherwise specified, we will assume our graphs to have planar or outerplanar embeddings. 
An edge of an embedded planar graph $G$ is \emph{outer} if it belongs to the outer face, and it is \emph{inner} otherwise. The \emph{weak dual} $\overline{G}$ of~$G$ is the graph having a node for each inner face of $G$, and an edge between two nodes if and only if the two corresponding faces share an edge. If $G$ is outerplanar, its weak dual $\overline{G}$ is a tree. If $\overline{G}$ is a path, $G$ is an \emph{outerpath}. A \emph{fan} is an internally-triangulated outerpath whose inner edges all share an end-vertex. 
%

An \emph{$st$-digraph} is a directed acyclic graph with a single source $s$ and a single sink $t$; an \emph{$st$-outerplanar graph} (resp.\ \emph{$st$-outerpath}) is an $st$-digraph whose underlying undirected graph is an outerplanar graph (resp.\ an outerpath). An \emph{$st$-fan} is an $st$-digraph whose underlying graph is a fan and whose inner edges have $s$ as an end-vertex. An $st$-outerplanar graph such that the edge $(s,t)$ exists is \emph{one-sided} if  $(s,t)$ is an outer edge, it is \emph{two-sided} if $(s,t)$ is an inner edge.

We recall a decomposition of $st$-outerpaths defined in~\cite{DBLP:journals/ejc/BhoreLMN23}.
The \emph{extreme faces} of an $st$-outerpath $G$ are the two faces that correspond to the two degree-one nodes of the weak dual $\overline{G}$. An $st$-outerpath $G$ is \emph{primary} if and only if one of its extreme faces is incident to $s$ and the other one to $t$. Observe that this definition is stronger than the one used in~\cite{DBLP:journals/ejc/BhoreLMN23}, in the sense that a primary $st$-outerpath according to our definition is a primary $st$-outerpath also according to the definition in~\cite{DBLP:journals/ejc/BhoreLMN23} (but the converse may not be true). 
Let $G$ be a primary $st$-outerpath. Consider a subgraph $F$ of $G$ that is an $xy$-fan (for some vertices $x,y$ of $G$). Let $\langle f_1,\dots,f_h \rangle$ be the list of faces forming the path $\overline{G}$ ordered from $s$ towards $t$. Note that the subgraph $F$ of $G$ is formed by a subset of  faces that are consecutive in the path $\langle f_1,\dots,f_h \rangle$. Let $f_i$ be the face of $F$ with the highest index, with $1 \le i \le h$. We say that $F$ is \emph{incrementally maximal} if $i=h$ or $F \cup f_{i+1}$ is not an $xy$-fan. For every face $f_i$ we denote by $\mi(f_i)$ the unique vertex of $f_i$ with one incoming edge and one outgoing edge in the boundary~of~$f_i$.

\begin{definition}
An \emph{$st$-fan decomposition} of a primary $st$-outerpath $G$ is a sequence of $s_i t_i$-fans $F_i \subseteq G$, with $i=1,\dots,k$, such that: (i) $F_i$ is incrementally maximal for each $i = 1, \dots, k$;  (ii) for any $1 \le i < j \le k$, $F_i$ and $F_j$ do not share any edge if $j>i+1$, while $F_i$ and $F_{i+1}$ share a single edge, which we denote by $e_i$; (iii) $s_1=s$; (iv) the tail of $e_i$ is $s_{i+1}$ for each $i=1, \dots, k-1$; (v) $e_i \neq (s_i,t_i)$ for each $i=1, \dots, k-1$;  and (vi) $\bigcup_{i=1}^k F_i=G$. Refer to Fig.~\ref{fi:fan_decomposition}.a-b.
\end{definition}

\begin{lemma}[\cite{DBLP:journals/ejc/BhoreLMN23}]
Every primary $st$-outerpath $G$ admits an $st$-fan decomposition.
\end{lemma}


Let $G$ be an $st$-outerplanar graph and let $\overline{P}$ be a path in the weak dual $\overline{G}$ of $G$ whose two endpoints are such that one corresponds to a face containing $s$ and the other one to a face containing $t$. Observe that the primal graph $G_{core}$ of $\overline{P}$ is a primary $st$-outerpath by construction, we call it the \emph{core} of $G$. On the other hand, if an outer edge $(u,v)$ of $G_{core}$ is not an outer edge of $G$, then it corresponds to a separation pair in $G$. In particular, $(u,v)$ belongs to $G_{core}$ and to another subgraph $A_{uv}$ of $G$ which is a one-sided $uv$-outerplanar graph; we call $A_{uv}$ an \emph{appendage} of $G$ \emph{attached} to $(u,v)$; refer to Fig.~\ref{fi:fan_decomposition}.c.

\begin{figure}[t]
\centering
\includegraphics[width=0.9\textwidth, page=1]{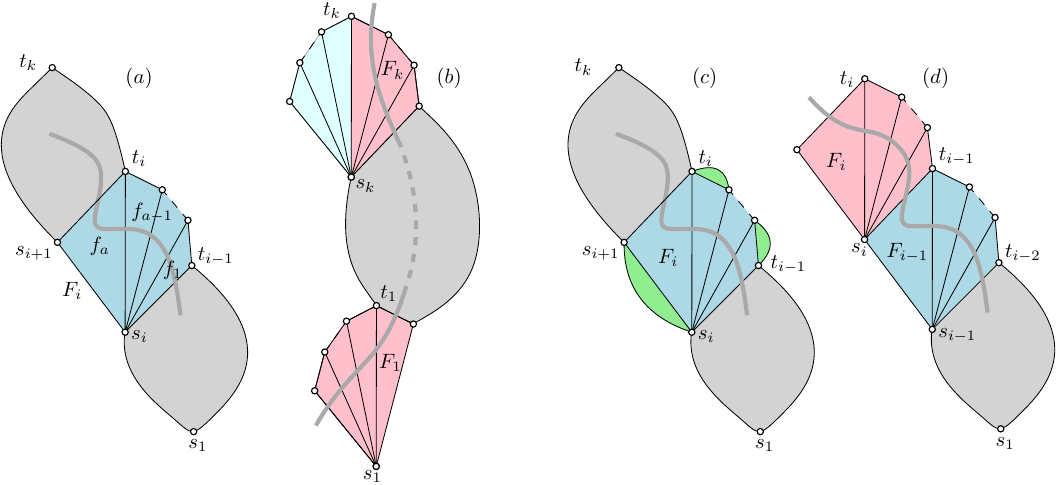}
\caption{(a) An $st$-fan in the middle of the $st$-fan decomposition.  (b) The first and last $st$-fans of the $st$-fan decomposition. Since the outerpath is primary, the last fan can always be chosen to be one-sided with edge $(s_k,t_k)$ regarded as a possible attachment edge of an appendage (light blue). (c) $G_{core}$ is blue and gray, while green subgraphs represent appendages of $G_{core}$. (d) Illustration of Property~\ref{lemma:properties_outerpl}.b. }
\label{fi:fan_decomposition}
\end{figure}

\begin{property2} \label{lemma:properties_outerpl}
Let $G$ be an $st$-outerplanar graph and let $G_{core}$ be the core of $G$. The following properties hold:
\begin{itemize}

\item[(a)] Every outer edge of $G_{core}$ is potentially an attachment edge of an appendage. 
 
\item[(b)] Let $F_1,\dots,F_k$ be an $st$-fan decomposition of $G_{core}$ and let $P$ be its dual path. Let $s_i$, $t_i$ denote the source and the sink of $F_i$. The fans $F_{i-1}$ and $F_i$ share the edge $(s_{i},t_{i-1})$; See Fig.~\ref{fi:fan_decomposition}.d (Stronger version of Lemma 3 in~\cite{DBLP:journals/ejc/BhoreLMN23}).

\item[(c)] Path $P$ enters $F_i$, $i=2,\dots,k$ through the edge $(s_{i},t_{i-1})$ and leaves $F_i$, $i=1,\dots,k-1$ through the edge $(s_{i+1},t_i)$.
\item[(d)] Let $F_i$ be a two-sided $st$-outerpath and let $f_1,\dots,f_a$ be the faces of $F_i$ as visited by $P$. Faces $f_1,\dots,f_{a-1}$, $a\geq 2$, lie on one side of $(s_i,t_i)$ and only the face $f_a$ lies on the other side of $(s_i,t_i)$. Refer to Fig.~\ref{fi:fan_decomposition}.a. 
\end{itemize}
\end{property2}

Property~(a) holds by definition. If Properties~(b) and (c) do not hold, then $G_{core}$ has either more than one sink or more than one source. Finally, assuming that Property~(d) does not hold, implies that $G_{core}$ is not an outerpath.

In this section we utilize a tool, called \emph{Hamiltonian completion}, that is another way to look at 2UTBEs. 
An upward planar graph $G$ has a 2UBE if and only if it is \emph{subhamiltonian}, i.e., it is a spanning subgraph of an upward planar $st$-digraph $\tilde{G}$ that has a directed Hamiltonian $st$-path~\cite{MchedlidzeS09}. 
More generally, there is an analogy between upward topological book embeddings and a more general form of  subhamiltonicity. Let $G$ be an upward planar graph and  $\tilde{G}=(V,\tilde{E})$ be an embedded $st$-digraph such that: (1) $G=(V,E)$ is a spanning subgraph of $\tilde{G}$, (2) $\tilde{G}$  has a directed Hamiltonian $st$-path $H$, and (3) each edge in $E$ is crossed by at most one edge of $\tilde{E}\setminus E$.  We say that $H$ is a \emph{subhamiltonian path} of $G$ and $\tilde{G}$ is an \emph{HP-completion} of $G$. See Fig.~\ref{fi:fan_proof} for an example of subhamiltonian paths.

\begin{lemma}[\cite{MchedlidzeS09}]\label{le:hp-completion}
An upward planar graph has a 2UTBE with at most one spine-crossing per edge if and only if it has an HP-completion.
The order of the vertices along the spine in the 2UTBE is the same as in the subhamiltonian~path.
\end{lemma}

The subhamiltonian path crosses some edges of $G$ by splitting them into \emph{sub-edges}. We inherit the definition of nesting (sub-)edges from 2UTBE to  HP-completion. Thus, the (sub-)edges $(u,v)$, $(w,z)$ \emph{nest} in $\tilde{G}$ if in the embedding of $\tilde{G}$ they are on the same side of the path $H$ and  $u \prec w \prec z \prec v$ or $w \prec u \prec v \prec z$. Since the order of the vertices on the spine of the book and along the Hamiltonian path coincide, two (sub-)edges  nest in $\tilde{G}$ if and only if they nest in the corresponding 2UTBE.  We now prove the key result of this section.

\begin{restatable}[{$\star$}]{lemma}{lestouterpath}\label{lemma:st-outerpath}
Every primary st-outerpath has an HP-completion without nesting sub-edges.      
\end{restatable}
\begin{proof}[sketch]
Let $G$ be a primary $st$-outerpath and $F_1,\dots,F_k$ be its $st$-fan decomposition. Let $\overline{P}$ be the dual path of $G$. 
Let $G_i$ be the subgraph of $G$ composed by $F_1,\dots,F_i$, $i=1,\dots,k$, therefore $G=G_k$. We construct the subhamiltonian path $H_i$ in $G_i$ by induction on $i$, assuming the next invariants for $H_{i-1}$ in $G_{i-1}$:

\begin{enumerate}
\item[$\mathcal{I}1$] Subhamiltonian path $H_{i-1}$ in $G_{i-1}$ terminates with the edge $(s_i,t_{i-1})$.
\item[$\mathcal{I}2$] Path $H_{i-1}$ crosses the edge $(s_{i-1},t_{i-1})$ (in a point referred to as $p_{i-1}$) if and only if $F_{i-1}$ is two-sided. No other edge of $F_{i-1}$ is crossed by $H_{i-1}$.
\item[$\mathcal{I}3$] $H_{i-1}$ does not create nesting sub-edges in $G_{i-1}$.  
\end{enumerate}

\noindent We show how to construct $H_i$ so to maintain the invariants. We have three cases based on whether $F_{i-1}$ and $F_i$ are two-sided or not. 

\noindent \paragraph{Case 1:} \textbf{both $F_{i-1}$ and $F_i$ are two-sided.}  Refer to Fig.~\ref{fi:fan_proof}.a.
Consider the dual path $\overline{P}$ in $F_i$ and let $f_1,\dots,f_a$, be the faces of $F_i$ as visited by $\overline{P}$. 
By Property~\ref{lemma:properties_outerpl}(d), since $F_i$ is  two-sided, faces $f_1,\dots,f_{a-1}$, $a\geq 2$,  lie on one side of $(s_i,t_i)$ and only the face $f_a$ lies on the other side of $(s_i,t_i)$.
Note that, by Properties~\ref{lemma:properties_outerpl}(b) and~\ref{lemma:properties_outerpl}(c),  $F_{i-1}$ and $F_i$ share $(s_{i},t_{i-1})$ and $P$ enters $F_i$ through $(s_{i},t_{i-1})$; it follows that $\mi(f_1)=t_{i-1}$. 
By induction hypothesis $H_{i-1}$ terminates at $(s_{i},t_{i-1})$. Therefore, we can set path $H_i$ to be $H_{i-1}$ concatenated with $\mi(f_1),\dots,\mi(f_a),t_i$. Note that $\mi(f_a)=s_{i+1}$, thus Invariant $\mathcal{I}1$ holds. 
Also, $H_i$ crosses $(s_i,t_i)$ and no other edge of $F_i$, hence Invariant $\mathcal{I}2$ holds as well.
Finally, concerning the only two edges of  $F_{i-1}$ and $F_i$ that are crossed by $H_i$, the order in which their end-vertices $s_i$, $t_i$, $s_{i-1}$, and $t_{i-1}$ and their crossing points $p_{i-1}$ and $p_{i}$ are visited is $s_{i-1},p_{i-1},s_i,t_{i-1},p_i,t_i$, which implies that their sub-edges do not nest. No other sub-edge is created by $H_i$, thus $\mathcal{I}3$ holds. 
\paragraph{Case 2:}  \textbf{$F_{i-1}$ is one-sided and $F_i$ is two-sided.}  Again, by Property~\ref{lemma:properties_outerpl}(c),  $H_{i-1}$ ends with $(s_{i},t_{i-1})$. Consider the dual path $\overline{P}$ in $F_i$ and let $f_1,\dots,f_a$, be the faces of $F_i$ as visited by $P$. 
By Property~\ref{lemma:properties_outerpl}(d), since $F_i$ is a two-sided, faces $f_1,\dots,f_{a-1}$, $a\geq 2$ lie on one side of $(s_i,t_i)$, and only the face $f_a$ lies on the other side of $(s_i,t_i)$. Faces $f_1,\dots,f_{a-1}$ lie on the same side of $(s_i,t_i)$ where $(s_i,t_{i-1})$ lies, as $\overline{P}$ enters $F_i$ through $(s_i,t_{i-1})$. Therefore, we can concatenate $H_{i-1}$ with $\mi(f_1)=t_{i-1},\dots,\mi(f_a),t_i$; see Fig.~\ref{fi:fan_decomposition}.b. Also in this case $\mi(f_a)={s_{i+1}}$ and therefore Invariant $\mathcal{I}1$ holds. Further, path $H_i$ crosses $(s_i,t_i)$ and Invariant $\mathcal{I}2$ also holds. Since $H_{i-1}$ has no nesting sub-edges, and any sub-edge in $G_{i-1}$ ends before the sub-edge of $F_i$ starts (since $t_{i-2}$ is connected by a directed path to $s_i$) we have that $H_i$ does not have nesting sub-edges and thus also $\mathcal{I}3$ holds.  

\noindent\paragraph{Case 3:}  \textbf{$F_i$ is one-sided.}  We distinguish two sub-cases depending on whether $t_{i-1}$ coincides with $t_{i}$ or not.

\noindent \paragraph{Case 3.a:} $t_{i-1} \neq t_i$ (refer to
Fig.~\ref{fi:fan_decomposition}.c-d). Since $H_{i-1}$ ends with $(s_{i,t_{i-1}})$, we can concatenate $H_{i-1}$ with $\mi(f_1)=t_{i-1},\dots,\mi(f_a),t_i$.  Also in these cases $\mi(f_a)=s_{i+1}$ and Invariant $\mathcal{I}1$ holds. Since $F_i$ is one-sided, Invariant $\mathcal{I}2$ trivially holds. Finally, since $F_i$ is one-sided, all the vertices $\mi(f_1)=t_{i-1},\dots,\mi(f_a),t_i$ lie on the same side as $(s_i,t_{i-1})$ and therefore path $H_i$ does not cross any edge of $F_i$. This, and the induction hypothesis $\mathcal{I}3$ imply that $H_i$ does not have nesting sub-edges.

\noindent \paragraph{Case 3.b:} $t_{i-1} = t_i$ (refer to
Fig.~\ref{fi:fan_decomposition}.e).In this case $a=1$, i.e., $F_i$ consists of a single triangle $s_{i}, \mi(f_1), t_i$. We remove the edge $(s_{i},t_{i-1})$ from $H_{i-1}$ and extends the resulting path with the two edges $(s_{i},\mi(f_1))$ and $(\mi(f_1),t_i)$. Invariant $\mathcal{I}1$ holds by construction; Invariant $\mathcal{I}2$ trivially holds because $F_i$ is one-sided; finally, Invariant $\mathcal{I}3$ holds because it held for $H_{i-1}$ and the extension to $H_i$ cannot create any nesting.
\end{proof}

\begin{figure}[t]
\centering
\includegraphics[width=0.9\textwidth, page=3]{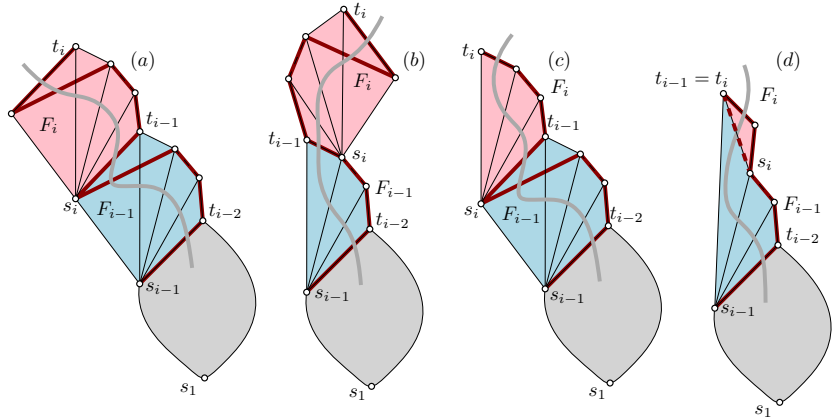}
\caption{Proof of \Cref{lemma:st-outerpath}: (a) Case 1. (b-d) Case 2, 3.a, and 3.b. The subhamiltonian path $H_i$ is drawn in dark red.}
\label{fi:fan_proof}
\end{figure}


\begin{lemma}\label{le:outerplanar}
Every st-outerplanar graph has an HP-completion without nesting sub-edges. 
\end{lemma}
\begin{proof}
Let $G$ be an $st$-outerplanar graph and let $G_{core}$ be the core of $G$. By Lemma~\ref{lemma:st-outerpath}, $G_{core}$  has an HP-completion with subhamiltonian path $H'$ that does not create nesting sub-edges. 
By Property~\ref{lemma:properties_outerpl}(a), every outer edge of $G_{core}$ is potentially an attachment edge of an appendage of $G$. We expand the subhamiltonian path $H'$ of $G_{core}$ to a subhamiltonian path $H$ in $G$ as follows, refer to Fig.~\ref{fi:appendages}. Let $A$ be an appendage of $P$ attached to an edge $e$ and let $f$ be the internal face of $G_{core}$ incident to the edge $e$. We flip $A$ to lie inside $f$. We visit all the vertices of $A$ that are not the source or the sink of $A$ either immediately after $H'$ visits the source of $A$ (blue appendage in Fig.~\ref{fi:appendages}.c) or immediately before it visits the sink of $A$ (pink appendage in Fig.~\ref{fi:appendages}.c), or both things at the same time (green appendages in Fig.~\ref{fi:appendages}.c).  After this procedure the edges crossed by $H$ are exactly the edges of $G_{core}$ crossed by $H'$, i.e., no new sub-edge is created. Further, the vertices of $G_{core}$ are visited in the same order by $H$ and by $H'$. Hence, since $H'$ did not create nesting sub-edges in $G_{core}$, so does $H$ in $G$. 
\end{proof}

\begin{figure}[t]
\centering
\includegraphics[width=0.7\textwidth, page=4]{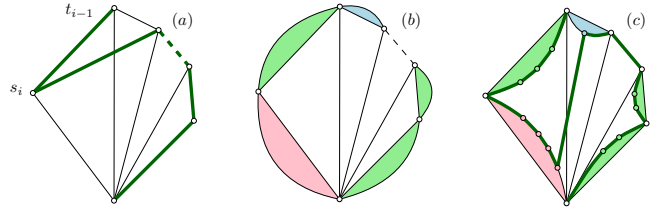}
\caption{Augmenting the subhamitonian path to visit appendages.}
\label{fi:appendages}
\end{figure}

By \Cref{le:outerplanar,le:hp-completion} every $st$-outerplanar graph has a nice 2UTBE with at most one spine-crossing per edge. By \Cref{le:UPSE-sufficiency} we have the following.

\begin{theorem}\label{th:st-outerplanar}
    Every $st$-outerplanar graph admits a $1$-bend UPSE on every UOSC point~set.
\end{theorem}

\section{1-bend UPSE are not always possible}\label{se:counterexample}

In this section we describe a $2$-outerplanar $st$-digraph $G$ and an UOSC point set $S$ such that $G$ does not admit a $1$-bend UPSE on $S$. An $st$-digraph is $2$-outerplanar if removing all vertices of the outer face yields an outerplanar digraph. The point set $S$ used in the proof of \Cref{le:counter-ex} is constructed in such a way that, assuming the existence of a 1-bend UPSE $\Gamma$ on $S$, no matter what is the single-top 2UTBE induced by $\Gamma$, there is always a forbidden configuration that is mapped to an impossible point set, thus implying the existence of a crossing. It is possible to define many point sets that have this property. The one shown in \Cref{fi:counter-example-1} has points with the following 
exact coordinates:

\begin{flalign*}
	&&p1&&&=&&(93,0), 
	&p2&&&=&&(79,2), 
	&p3&&&=&&(73,4), \\
	&&p4&&&=&&(16,43),
	&p5&&&=&&(8,49),
	&p6&&&=&&(4,53), \\
	&&p7&&&=&&(2,56),
	&p8&&&=&&(1,59),
	&p9&&&=&&(0,63), \\
	&&p10&&&=&&(0,67),
	&p11&&&=&&(1,71),
	&p12&&&=&&(2,74), \\
	&&p13&&&=&&(4,77),
	&p14&&&=&&(8,81),
	&p15&&&=&&(16,87), \\
	&&p16&&&=&&(73,126),
	&p17&&&=&&(79,128),
	&p18&&&=&&(93,130).
\end{flalign*}

\begin{figure}[t]
    \centering
    \subfigure[]{\label{fi:counter-example-a}\includegraphics[width=0.44\textwidth, page=1]{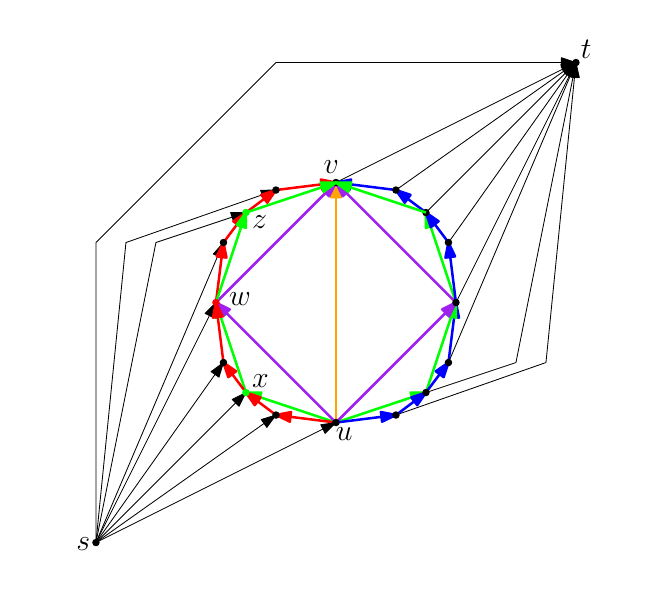}}
    \hfil
    \subfigure[]{\label{fi:counter-example-b}\includegraphics[width=0.44\textwidth, page=2]{figures/counter-example.pdf}}
    \caption{(a) An $st$-digraph $G$ and (b) an UOSC point set $S$ for the proof of \cref{le:counter-ex}}
    \label{fi:counter-example-1}
\end{figure}

\begin{restatable}[{$\star$}]{lemma}{counterEx}\label{le:counter-ex}
There exists a $2$-outerplanar $st$-digraph $G$ and an UOSC point set $S$ such that $G$ does not admit a $1$-bend UPSE on $S$.
\end{restatable}
\begin{proof}
Let $G$ be the $st$-digraph of \Cref{fi:counter-example-a} and let $S$ be the point set of \Cref{fi:counter-example-b} 
By \Cref{le:UPSE-necessity}, if $G$ has a 1-bend UPSE $\Gamma$ on $S$, then $\Gamma$ induces a single-top  2UTBE.  We show that every single-top 2UTBE $\gamma$ of $G$ has a forbidden configuration of Type $i$, for some $i \in \{1,2,3,4\}$, that is necessarily mapped to a Type $i$ impossible point subset of $S$. By \Cref{le:forbidden} a 1-bend UPSE cannot exist. Let $p_1, p_2,\dots, p_{18}$ be the points of $S$ in bottom-to-top~order.
%
Let $\pi_{l}$ be the path from $u$ to $v$ to the left of $(u,v)$ (red in \Cref{fi:counter-example-a}) and let $\pi_r$ be the path from $u$ to $v$ to the right of $(u,v)$ (blue in \Cref{fi:counter-example-a}). The edge $(u,v)$ (yellow in \Cref{fi:counter-example-a}) has vertices on both sides. Thus, in every 2UTBE it crosses the spine either once or twice and the vertices of $\pi_l$ must appear along the spine in the order they appear along $\pi_l$; the same holds for $\pi_r$. We have different cases. In each case we denote by $v_1, v_2, \dots,v_n$ the sequence of vertices along the spine (thus vertex $v_i$ is mapped to point $p_i$). In all cases $u$ is mapped to $p_2$ and $v$ is mapped to $p_{17}$.

\begin{figure}[t]
    \centering
    \subfigure[]{\label{fi:counter-example-c}\includegraphics[width=0.44\textwidth, page=3]{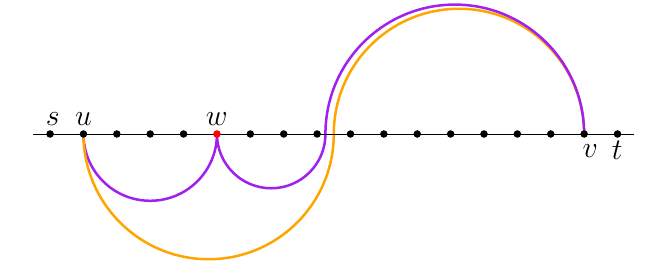}}
    \subfigure[]{\label{fi:counter-example-d}\includegraphics[width=0.44\textwidth, page=4]{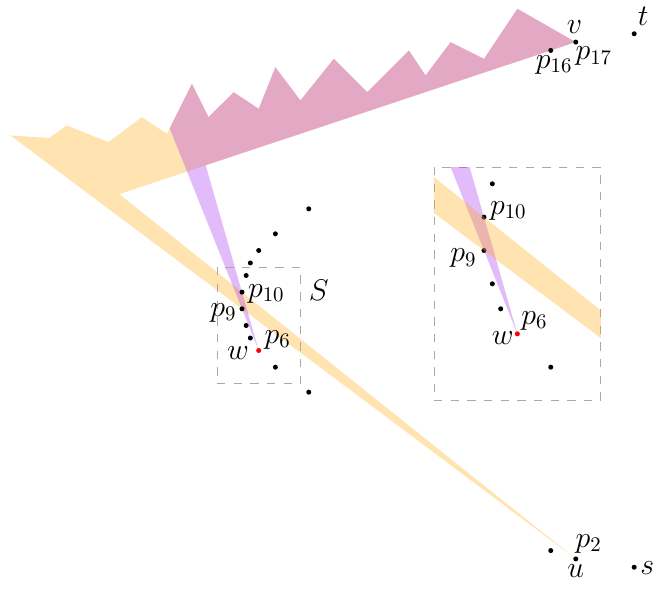}}
    \subfigure[]{\label{fi:counter-example-e}\includegraphics[width=0.48\textwidth, page=5]{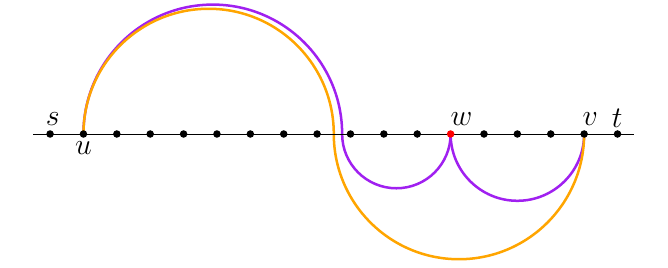}}
    \subfigure[]{\label{fi:counter-example-f}\includegraphics[width=0.48\textwidth, page=6]{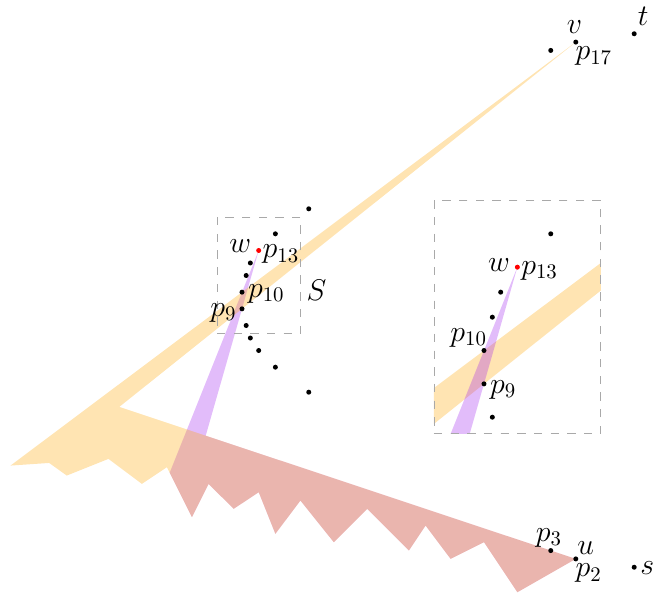}}
    \caption{Case 1 of \Cref{th:counter-ex}.}
    \label{fi:counter-example-2}
\end{figure}

\begin{figure}[t!]
    \centering
    \subfigure[]{\label{fi:counter-example-g}\includegraphics[width=0.48\textwidth, page=7]{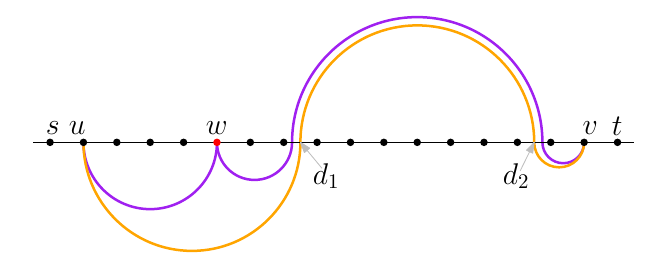}}
    \subfigure[]{\label{fi:counter-example-h}\includegraphics[width=0.48\textwidth, page=8]{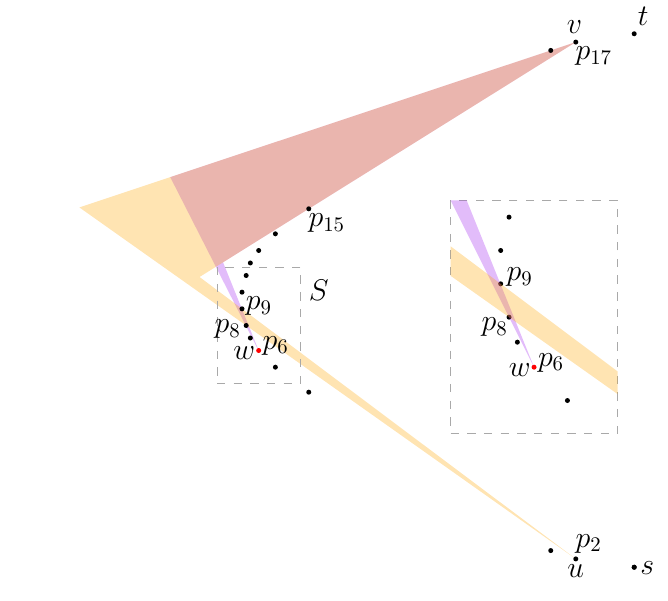}}
    \caption{Case 2.A of \Cref{th:counter-ex}.}
    \label{fi:counter-example-3}
\end{figure}

\noindent \textbf{Case 1:} Edge $(u,v)$ crosses the spine once (see \cref{fi:counter-example-c,fi:counter-example-e}). The first sub-edge of $(u,v)$ is either a bottom or a top sub-edge. Int the first (resp. second) case the vertex $w$ coincides with $v_6$ (resp. $v_13$) and the edges $(w,v)$ and $(u,v)$ form a Type 2 (resp. a Type 4) forbidden configuration with the spine crossings between $v_9$ and $v_{10}$. Since, $p_2,p_6,p_9,p_{10},p_{16},p_{17}$ (resp. $p_2,p_3,p_9,p_{10},p_{13},p_{17}$) form a Type 2 (resp. a Type 4) impossible point set (see \Cref{fi:counter-example-d,fi:counter-example-f}), by \Cref{le:forbidden} a 1-bend UPSE cannot exist in this case.

\noindent \textbf{Case 2:} Edge $(u,v)$ crosses the spine twice. In this case $(u,v)$ consists of three sub-edges $(u,d_1)$, $(d_1,d_2)$, and $(d_2,v)$, where $d_1$ and $d_2$ are spine crossings. Only $(d_1,d_2)$ is a top sub-edge. Thus, the vertices of $\pi_l$ have to be distributed in the two intervals defined by $(u,d_1)$ and $(d_2,v)$. We distinguish six sub-cases (five are omitted)  depending on the distribution of the vertices of $\pi_l$.

\noindent \textit{Case 2.A:} $w$ is between $u$ and $d_1$ with a single vertex of $\pi_l$ between $d_2$ and $v$ (see \cref{fi:counter-example-g}). In this case $w$ coincides with $v_6$ and both $(w,v)$ and $(u,v)$ cross the spine between $v_8$ and $v_9$ and between $v_{15}$ and $v_{16}$. The edges $(w,v)$ and $(u,v)$ form a Type 2 forbidden configuration. Since $p_2,p_6,p_8,p_{9},p_{15},p_{17}$ form a Type 2 impossible point set (see \Cref{fi:counter-example-h}), by \Cref{le:forbidden} a 1-bend UPSE cannot exist.

\noindent \textit{Case 2.B:} Vertex $w$ is between $u$ and $d_1$ and there are two vertices of $\pi_l$ between $d_2$ and $v$ (see \cref{fi:counter-example-i}). In this case $w$ coincides with $v_6$ and both edges $(w,v)$ and $(u,v)$ cross the spine once between $v_7$ and $v_8$ and another time between $v_{14}$ and $v_{15}$.  Also in this case the edges $(w,v)$ and $(u,v)$ form a Type 2 forbidden configuration. Since $p_2,p_6,p_7,p_{8},p_{14},p_{17}$ form a Type 2 impossible point set (see \Cref{fi:counter-example-j}), by \Cref{le:forbidden} a 1-bend UPSE cannot exist in this case.

\begin{figure}[t]
	\centering
	\subfigure[]{\label{fi:counter-example-i}\includegraphics[width=0.48\textwidth, page=9]{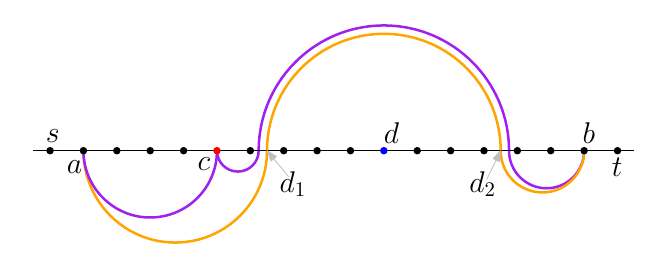}}
	\subfigure[]{\label{fi:counter-example-j}\includegraphics[width=0.48\textwidth, page=10]{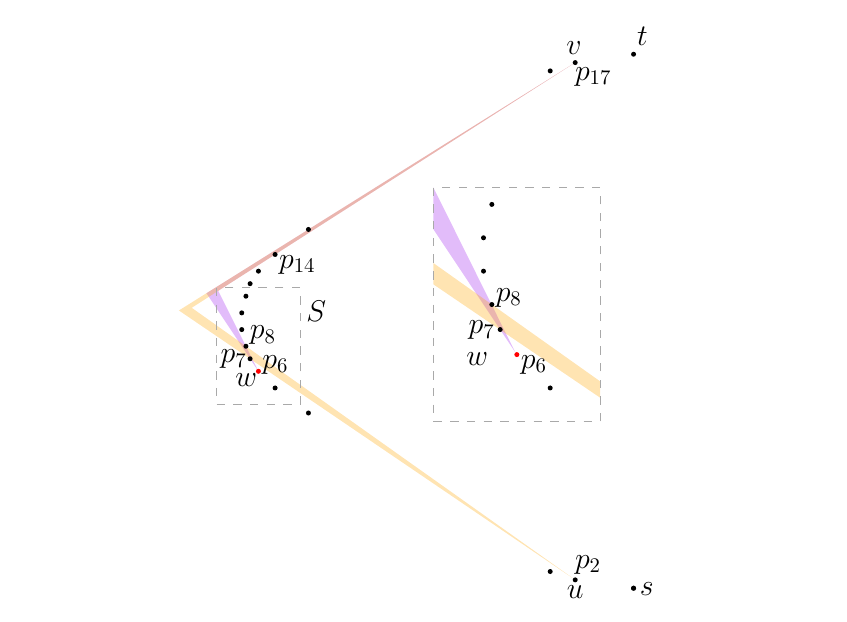}}
	\caption{Case 2.B of \Cref{th:counter-ex}.}
	\label{fi:counter-example-4}
\end{figure}

\noindent \textit{Case 2.C:} Vertex $w$ is between $u$ and $d_1$ and there are three vertices of $\pi_l$ between $d_2$ and $v$ (see \cref{fi:counter-example-k}). In this case $w$ coincides with $v_6$ and the vertex $z$ coincides with $v_{15}$. Edges $(w,v)$ and $(w,z)$ form a Type 4 forbidden configuration with spine crossings between $v_{13}$ and $v_{14}$. Since $p_6,p_13,p_{14},p_{15},p_{17}$ form a Type 4 impossible point set (see \Cref{fi:counter-example-l}), by \Cref{le:forbidden} a 1-bend UPSE cannot exist in this case.

\begin{figure}[h!]
	\centering
	\subfigure[]{\label{fi:counter-example-k}\includegraphics[width=0.48\textwidth, page=11]{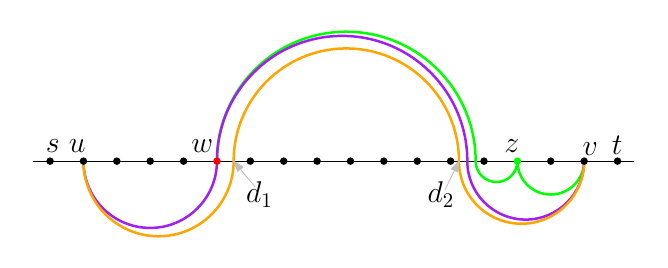}}
	\subfigure[]{\label{fi:counter-example-l}\includegraphics[width=0.48\textwidth, page=12]{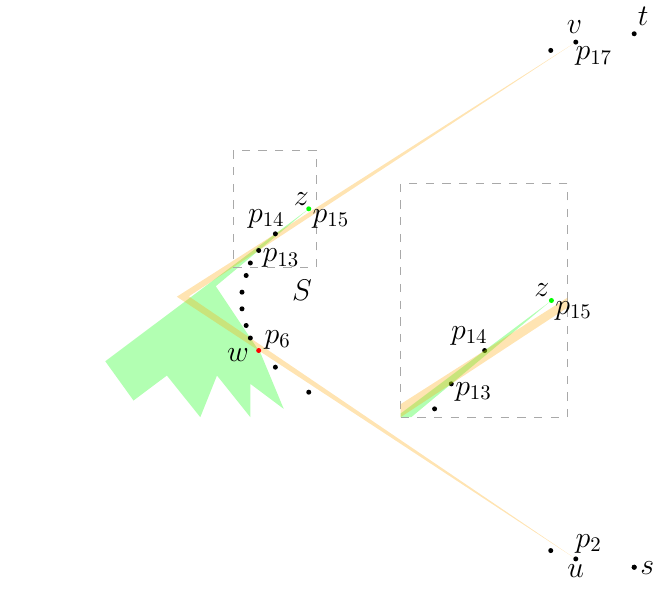}}
	\caption{Case 2.C of \Cref{th:counter-ex}.}
	\label{fi:counter-example-5}
\end{figure}

\noindent \textit{Case 2.D:} vertex $w$ is between $d_2$ and $v$ and there is only one vertex of $\pi_l$ between $u$ and $d_1$. The proof is symmetric to the Case 2.A (see \cref{fi:counter-example-m,fi:counter-example-n}).

\noindent \textit{Case 2.E:} vertex $w$ is between $d_2$ and $v$ and there are two vertices of $\pi_l$ between $u$ and $d_1$. The proof is symmetric to the Case 2.B (see \cref{fi:counter-example-o,fi:counter-example-p}).

\noindent \textit{Case 2.F:} vertex $w$ is between $d_2$ and $v$ and there are three vertices of $\pi_l$ between $u$ and $d_1$. The proof is symmetric to the Case 2.C (see \cref{fi:counter-example-q,fi:counter-example-r}).

\begin{figure}[htbp]
	\centering
	\subfigure[]{\label{fi:counter-example-m}\includegraphics[width=0.48\textwidth, page=13]{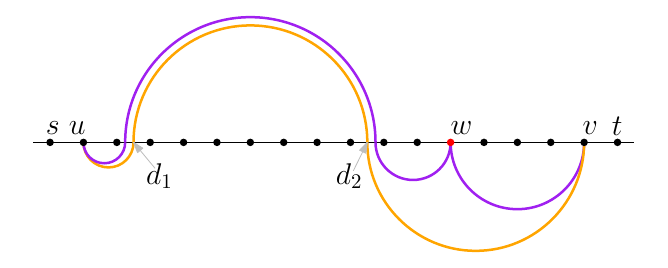}}
	\subfigure[]{\label{fi:counter-example-n}\includegraphics[width=0.48\textwidth, page=14]{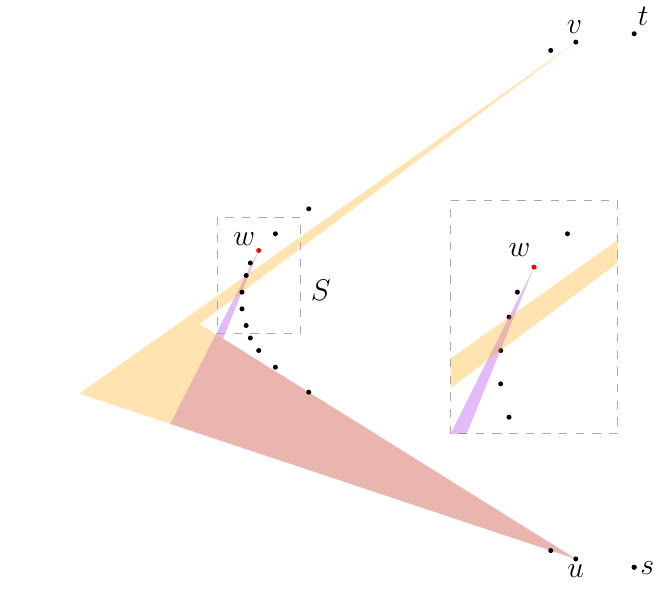}}
	\subfigure[]{\label{fi:counter-example-o}\includegraphics[width=0.48\textwidth, page=15]{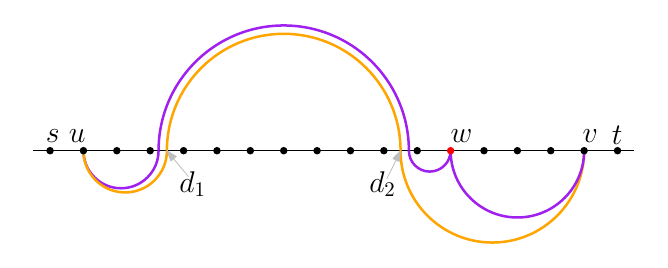}}
	\subfigure[]{\label{fi:counter-example-p}\includegraphics[width=0.48\textwidth, page=16]{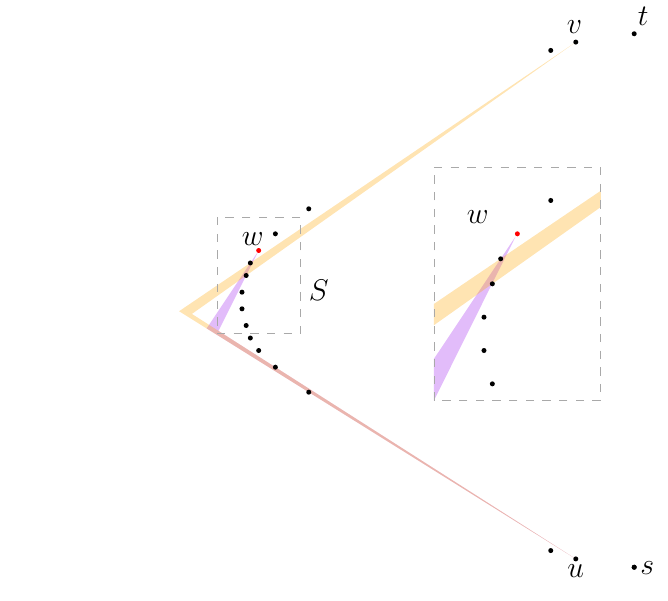}}
	\subfigure[]{\label{fi:counter-example-q}\includegraphics[width=0.48\textwidth, page=17]{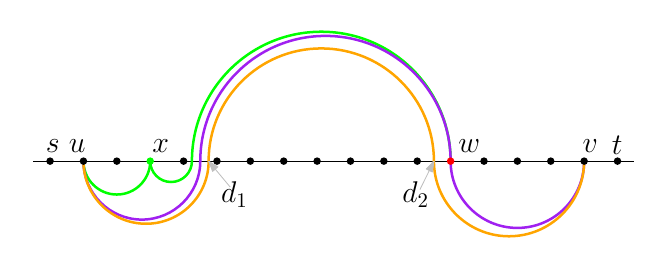}}
	\subfigure[]{\label{fi:counter-example-r}\includegraphics[width=0.48\textwidth, page=18]{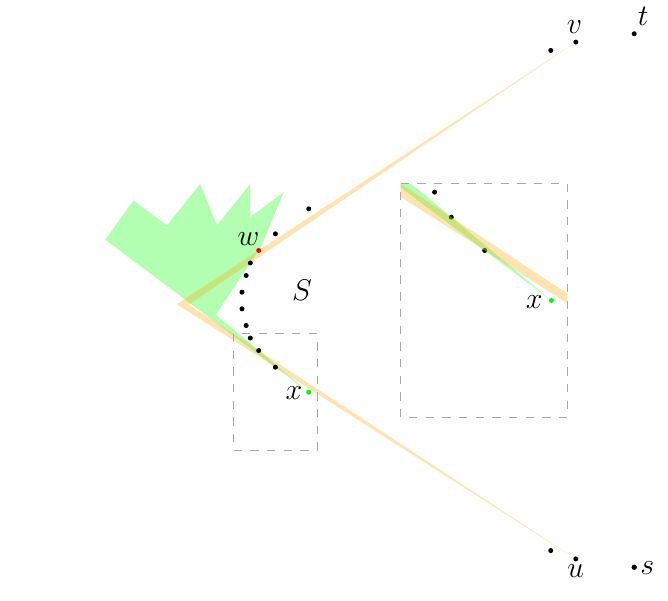}}
	\caption{Case 2.D--2.F of \Cref{th:counter-ex}.}
	\label{fi:counter-example-6}
\end{figure}

\end{proof}

The following theorem is easily derived from \Cref{le:counter-ex} by suitably adding, for every $n \geq 18$, $n-18$ vertices to $G$ and $n-18$ points to $S$.

\begin{restatable}{theorem}{thcounterex}\label{th:counter-ex}
For every $n \geq 18$ there exists an $n$-vertex $2$-outerplanar $st$-digraph $G$ and an UOSC point set $S$ such that $G$ does not admit a $1$-bend UPSE on $S$.    
\end{restatable}
\begin{proof}
	For every $n > 18$ we can transform the $st$-digraph $G$ of \Cref{le:counter-ex} to an $st$-digraph $G'$ with $n$ vertices as follows. We add to $G$ a directed path $s_0 s_1 \dots s_{n'+1}$, with $n'=n-18$, so that $s_{n'+1}$ coincides with the single source $s$ of $G$; we then connect each $s_i$ to the vertices $u$ and $t$ of $G$, for  $i=1,2,n'$. The resulting digraph is an $st$-digraph. We also transform the point set $S$ into a point set $S'$, by adding $n'$ points that form a one-sided convex point set together with $S$ and such that the added points are the lowest of $S'$. It is easy to see that if $G'$ admits a $1$-bend UPSE of $S'$ then $G$ admits a $1$-bend UPSE on $S$.      
\end{proof}

\section{Open problems}\label{se:open}

Various questions  remain open related to \Cref{pr:one,pr:two} of \Cref{se:intro}. Among~them: 
%
\begin{enumerate}
    \item Investigate the non-upward version of \Cref{pr:one}. We observe that the graph of \Cref{th:counter-ex} is not a counterexample for this problem as it admits a (non-upward) PSE on the set of points $S$ of \Cref{th:counter-ex} (see \Cref{fi:drawing-non-upward-of-counterexample}).
    \item Study \Cref{pr:two}. In particular, it would be nice to find a characterization of the digraphs that admit a $1$-bend UPSE on every UOSC point set.   
\end{enumerate}


\begin{figure}[h!]
	\centering
	\includegraphics[width=0.85\textwidth, page=1]{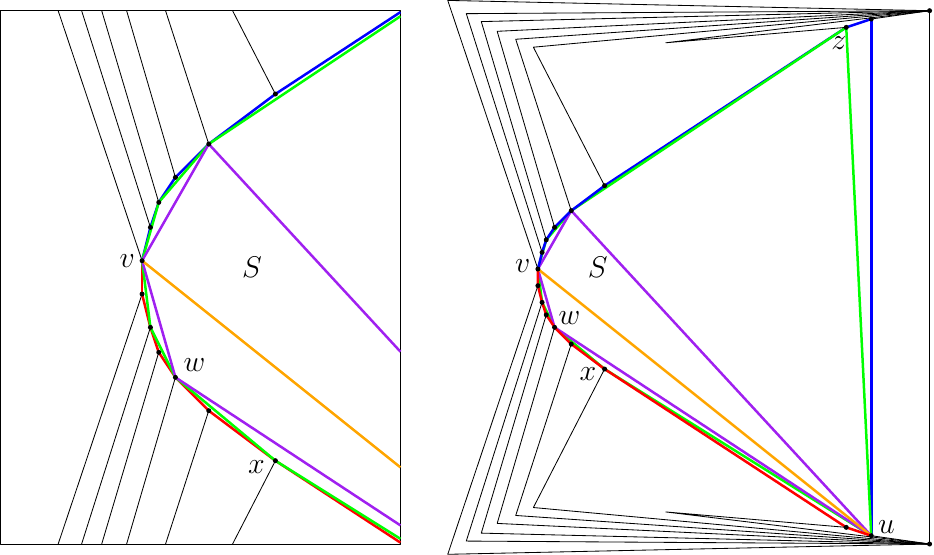}
	\caption{A (non-upward) PSE of the graph~of~\Cref{fi:counter-example-a} on the point set of~\Cref{fi:counter-example-b}.}
	\label{fi:drawing-non-upward-of-counterexample}
\end{figure}


\bibliography{biblio}
\bibliographystyle{splncs04.bst}
\end{document}